\documentclass[envcountsect, fleqn, 11pt]{llncs}

\usepackage[left=1in,right=1in,top=1in,bottom=1in]{geometry}
\usepackage{amsmath}
\usepackage{amsfonts}
\usepackage{enumerate}
\usepackage{graphicx}
\usepackage{algorithm}
\usepackage{algorithmicx}
\usepackage{algpseudocode}
\usepackage{verbatim}
\usepackage{subfigure}
\usepackage{multicol}
\usepackage{xcolor}

\newcommand\refDefinition[1]{\hyperref[#1]{Definition~\ref*{#1}}}
\numberwithin{equation}{section}

\renewcommand{\algorithmicrequire}{\textbf{Input:}}
\renewcommand{\algorithmicensure}{\textbf{Output:}}

\algdef{SE}[PROCEDURE]{Procedure}{EndProcedure}%
[2]{\algorithmicprocedure\ \textproc{#1}\ifthenelse{\equal{#2}{}}{}{(#2)}}%
{\algorithmicend\ \algorithmicprocedure}%

\algnewcommand\algorithmicforeach{\textbf{for each:}}
\algnewcommand\ForEach{\item[ \algorithmicforeach]}

\begin{document}
\title{Revenue Maximizing Envy-Free Pricing in Matching Markets with Budgets} 

\author{Riccardo Colini-Baldeschi \inst{1} \and Stefano Leonardi  \inst{2} \thanks{
This work is partly supported by the EU FET project MULTIPLEX no. 317532 and the Google Focused Award on ``Algorithms for Large-scale Data Analysis''.} \and Qiang Zhang \inst{2}}
\institute{Dipartimento di Economia e Finanza, LUISS, Rome,  \email{rcolini@luiss.it}
\and Sapienza University of Rome, \email{leonardi@dis.uniroma1.it}, \email{qzhang@gmail.com}}

\maketitle

\begin{abstract}

We study envy-free pricing mechanisms in matching markets with $m$ items and $n$ budget constrained buyers.  
Each buyer is interested in a subset of the items on sale, and she appraises at some single-value every item in her preference-set.  Moreover, each buyer has a budget that constraints the maximum affordable payment, while  
she aims to obtain as many items as possible of her preference-set.  
Our goal is to compute an envy-free pricing allocation that maximizes the revenue, i.e., the total payment charged to the buyers.  
This pricing problem is hard to approximate better than  $\Omega({\rm min} \{n,m\}^{1/2-\epsilon})$ for any $\epsilon>0$, unless $P=NP$~\cite{DBLP:conf/wine/Colini-BaldeschiLSZ14}.  The hardness result is due to the presence of the matching constraints given that the  simpler multi-unit case can be approximated up to a constant factor of $2$ \cite{DBLP:conf/sigecom/FeldmanFLS12}. 
The goal of this paper is to circumvent the hardness result by restricting ourselves to specific settings of valuations and budgets. Two particularly significant scenarios are: each buyer has a budget that is greater than her single-value valuation, and each buyer has a budget that is lower than her single-value valuation. Surprisingly, in both scenarios we are able to achieve a $1/4$-approximation to the optimal envy-free revenue. The algorithms utilize a novel version of the Ausebel \cite{Ausubel04} ascending price auction. These results may suggest that, although it is difficult to approximate the optimal revenue in general, ascending price auctions could achieve relatively good revenue in most of the practical settings.

\end{abstract}

\section{Introduction}

In this paper we study revenue maximization with envy-free pricing in matching markets. 
Imagine a seller that would like to sell $m$ different items to $n$ buyers.
Every buyer $i$ is interested in a subset of items $S_i$ (the \textit{preference-set}) and has a budget $b_i$ that represents her maximum affordable payment.
Moreover, every buyer $i$ appraises each item in her preference-set at value $v_{i}$, and any other item has zero-value for her.
The buyers are willing to get the largest number of items in their preference-sets. But every buyer $i$ doesn't want to pay more than her value $v_{i}$ for an item and more than her budget $b_{i}$ for the whole set of obtained items. 

The seller has full knowledge on the buyer types and aims to compute an outcome that maximizes her revenue.
The outcome is composed by a payment vector (a payment $p_{i}$ for each buyer $i$) and an allocation vector (a, potentially empty, set of items $X_{i}$ for each buyer $i$).

Every buyer $i$, given a payment $p_i$ and an allocation $X_i$, has a utility equal to $v_i \cdot |X_i \cap S_i| - p_i$ if the payment is less than or equal to the budget, i.e. $p_{i} \leq b_{i}$. However, the utility becomes $-\infty$ if the payment exceeds the budget, i.e. $p_{i} > b_{i}$.

Our goal is to design a pricing algorithm that is able to provide a good revenue to the seller and observe some fairness criterion for the buyers.
In order to model fairness, we consider two notions: (i) the \textit{individual rationality}, and (ii) the \textit{envy-freeness}.
Individual rationality means that at the end of the auction no buyer will experience a negative utility.
Envy-freeness means that, given a pricing scheme (an assignment of prices to items or sets of items), every buyer obtains the most desired set of items.

Different notions of envy-freeness have been studied in literature, and a detailed discussion about them is deferred to Section \ref{sec:relatedwork}.
In classical economics, the standard envy-free definition embodies {\em bundle pricing} ~\cite{Foley,Varian}. In the bundle pricing scheme each (different) set of items has a (potentially) different price. Thus in a \textit{bundle-price envy-free} allocation no buyer has an incentive to barter her bundle with the bundle of someone else.

More recently, envy-freeness has been often studied in the more restrictive setting of {\em item pricing} ~\cite{BalcanBM08,Briest:2006:SUS:1109557.1109678,DBLP:conf/approx/ChalermsookCKK12,CS08,GHKKKM05,DBLP:conf/soda/DemaineHFS06}.
In an item pricing scheme each (different) item has a (potentially) different price. Thus, the price for a bundle is obtained by summing the single items' prices contained in it.  Consequently, in \textit{item-price envy-free} allocations every buyer receives the bundle that maximizes her own utility. 

The difference is that the former definition of envy-freeness allows a buyer to envy only bundles
that are assigned to someone else.
Instead, the latter definition allows a buyer to be envious if the obtained bundle is not the utility-maximizer bundle (even if the utility-maximizer bundle is not allocated at all).
In this paper we will use both definitions, and we will explicitly state which envy-free condition is satisfied for each provided algorithm.
All our algorithms use as benchmark the optimal bundle-price envy-free revenue.\footnote{A bundle pricing scheme is able to extract more revenue than an item pricing scheme. Thus, competing against the optimal bundle-price envy-free revenue is the hardest task in this context. See Section \ref{sec:relatedwork}.}

The envy-free revenue-maximization problem with budget constraints has been initially studied by Feldman et al.~\cite{DBLP:conf/sigecom/FeldmanFLS12} in the multi-unit setting that considers one single kind of items (no matching constraints).  The authors discussed the limitations of the different pricing schemes, specifically they showed through separation examples that an item-pricing scheme cannot achieve more than $O(1/m)$ fraction of the optimal envy-free bundle-price revenue on some specific example. Thus, relying on a bundle pricing scheme, they provided a $2$-approximation algorithm to the optimal envy-free bundle-price revenue for multi-unit setting with budgeted buyers.

More recently, Branzei et al.~\cite{DBLP:journals/corr/BranzeiFMZ16} studied the same multi-unit setting in the context of item-pricing scheme, and they gave an FPTAS for the optimal envy-free item-price revenue (and an exact algorithm for special cases) using an item-pricing scheme.   

The envy-free revenue-maximization problem in matching markets with budget constraints has been considered for the first time in
Colini{-}Baldeschi et. al.~\cite{DBLP:conf/wine/Colini-BaldeschiLSZ14}. They proved that 
for $n$ buyers and $m$ items, the optimal  revenue cannot be approximated by a polynomial time algorithm within $\Omega({\rm min} \{n,m\}^{1/2-\epsilon})$ for any $\epsilon>0$, unless $P=NP$.  In this paper we present a novel approach that allows us to circumvent the $\Omega({\rm min} \{n,m\}^{1/2-\epsilon})$ impossibility result for practically relevant special cases.

\subsection{Our results}

We first consider the case where every buyer has a budget that is greater than or equal to her valuation. For example, in online advertising, it is typical that the advertisers' budgets are greater than their CPM valuations for a block of  impressions, and the advertisers are interested in buying multiple blocks of impressions from different publishers. Second, complementary to the first case, we study the case where every buyer has a budget that is less than her valuation.  

%
In the two cases we are able to provide:

\begin{itemize}

\item an individually-rational and item-price envy-free algorithm that achieves a 4-approximation to the optimal bundle-price envy-free revenue, when $b_{i} \geq v_{i}$ for every buyer $i$, and

\item an individually-rational and bundle-price envy-free algorithm that achieves a 4-approximation to the optimal bundle-price envy-free revenue, when $b_{i} < v_{i}$ for every buyer $i$.

\end{itemize}

The key ingredient needed to achieve these results is the design of an ascending price auction (inspired by the Ausubel auction \cite{Ausubel04}) for the matching market setting.
Our implementation adopts the standard ascending price auction's lifecycle: (i) the price is raised until some condition is met (usually referred as \textit{selling condition}),
(ii) when the selling condition is met an appropriate \textit{selling procedure} is executed, and the auction goes back to (i).
The main novelty introduced in our implementation is about the selling procedure.
The idea is to use proper arguments from matching theory ($B$-matchings and their properties) to compute an envy-free allocation of items to buyers.

Besides envy-freeness we are interested to secure the fairest pricing scheme to the buyers. The fixed-price scheme (one single price for all the items) is the standard
way to achieve fairness in setting with identical items. But, as expected, in setting with different items (like matching markets) a fixed-price scheme cannot obtain more than  a logarithmic fraction of the optimal revenue regardless of computational complexity considerations (see Appendix \ref{sec:lowerbound}).

Thus different prices for different items is a strict requirement to achieve a constant approximation. Nonetheless our aim is still to diversify the prices of the items as little as possible to avoid discrimination among buyers.

In the first case ($b_{i} \geq v_{i}$), if a buyer obtains some items, then every item she obtained is charged at the same price, and all the buyers that desire one of that items (and are able to pay for them) obtain other items at the same price. Intuitively, the idea is to partition the initial market in many submarkets and apply a fixed-price scheme in each of them.

Unfortunately, the same approach cannot be used in the second case ($b_{i} < v_{i}$).  Here it is necessary to adopt a bundle-pricing scheme. Namely, it is possible that different items are assigned at different prices to the same buyer, but for each price every buyer obtains the best bundle among the bundles assigned at the same price.
A bundle-pricing scheme is a strict requirement to obtain a constant approximation, indeed in Feldman et al. \cite{DBLP:conf/sigecom/FeldmanFLS12} a separation example showed that with an item-pricing scheme it is not possible to achieve an approximation factor better than $m$ (number of items) to the optimal envy-free revenue.\footnote{We remark that the separation example can be easily translated in an instance of the case in which every buyer $i$ has $b_{i}<v_{i}$. Instead, in the first case, when every buyer $i$ has $b_{i} \geq v_{i}$, the separation example fails. Indeed, in that case we are able to provide an item pricing scheme that obtains a constant approximation to the optimal bundle-price envy-free revenue.}

Since the ascending price auction approach can be adapted in either special cases.
We conclude that the general inapproximability result of this problem provided in Colini{-}Baldeschi et. al.~\cite{DBLP:conf/wine/Colini-BaldeschiLSZ14} is the result of the interaction between the above mentioned buyer classes.

\subsection{Related work}
\label{sec:relatedwork}

\textit{Envy-free pricing.} The notion of envy-freeness was initially introduced by Foley~\cite{Foley} and Varian~\cite{Varian}. The key property of an envy-free allocation is that no buyer has incentive
to exchange her bundle-payment pair with the bundle-payment pair of another buyer. That is to say, given a set of bundles and the corresponding prices, every buyer obtains
the bundle that maximizes her utility. This definition became the standard definition in the economics literature.
More recently,  the notion of envy-freeness has been deeply studied with a different perspective that involves item-pricing instead of bundle-pricing.
Indeed the key property of an envy-free allocation has been reshaped as follows: given a set of items and the corresponding (per-item) prices, every buyer obtains the bundle that maximizes her utility. Notice that the price of a bundle is automatically obtained from the sum of the items' prices contained in it.
It is easy to see that this definition is more difficult to satisfy than the classical economics definition.
The item-pricing version of envy-freeness has been considered in 
\cite{BalcanBM08,Briest:2006:SUS:1109557.1109678,DBLP:conf/approx/ChalermsookCKK12,CS08,GHKKKM05,DBLP:conf/soda/DemaineHFS06}.

The envy-freeness was studied in the context of different pricing schemes by Feldman et al.~\cite{DBLP:conf/sigecom/FeldmanFLS12}.
Specifically, they focused they attention on: (i) the bundle-pricing scheme, where it is possible to specify different prices for different bundles, (ii) the item-pricing scheme,
where different prices can be assigned to different items, and (iii) proportional-pricing scheme, that embody an item-pricing scheme plus the possibility to specify a maximum or a minimum size on the bundles that can be demanded by the buyers.
Moreover, the authors were able to rank these pricing schemes with respect to two distinct criteria: the customer experience (how much the customers judge fair a pricing scheme) or the obtainable revenue (how much a pricing scheme is able to extract from the customers).
In terms of customer experience the most desired pricing scheme is the item-pricing scheme, then the proportional-pricing scheme, and finally the bundle-pricing scheme.
This is because the customers prefer the pricing schemes that allow less discrimination and are more uniform (uniformity is perceived as fairness).
As expected, on the revenue perspective the ranking is reversed: the bundle-pricing scheme is able to produce the highest revenue, then proportional-pricing scheme is able to extract something less, and then the item-pricing scheme is the scheme that is able to extract less.
Feldman et al.~\cite{DBLP:conf/sigecom/FeldmanFLS12} have shown several separation examples that clearly state the limits of the different pricing schemes.

In our paper we aim to design algorithms that embody the fairest pricing scheme and compete against the powerful optimum. That is, we aim to design item-pricing algorithms that obtain a constant approximation to the optimal bundle-pricing envy-free revenue. When the item-pricing scheme fails to obtain a constant approximation, then we embody a proportional-pricing (or bundle-pricing) scheme.

\textit{Revenue-maximization.} Moreover, Feldman et al.~\cite{DBLP:conf/sigecom/FeldmanFLS12} proved that the problem of computing an optimal bundle-pricing envy-free revenue is NP-hard.
Thus, they provided
a $2$-approximation algorithm for multi-unit setting with budget constraints, it is also proved that this result is tight.
Colini-Baldeschi et al.~\cite{DBLP:conf/wine/Colini-BaldeschiLSZ14} investigated the problem in the context of multi-unit fixed-price auctions with budget constraints and matching markets with budget constraints.
Particularly relevant for our paper is the hardness result presented there. Indeed, they proved that the revenue maximizing envy-free problem in matching markets is
$\Omega({\rm min} \{n,m\}^{1/2-\epsilon})$ inapproximable for any $\epsilon>0$, unless $P=NP$.
This is why we are forced to study (particularly relevant) special cases.
The envy-free revenue maximization problem in multi-unit setting with budgets is also studied in \cite{DBLP:journals/corr/BranzeiFMZ16}. They provided algorithms that approximate optimal social welfare and optimal revenue with an item-pricing scheme when buyers are price takers, and an impossibility result for price making buyers.
The main difference between \cite{DBLP:conf/wine/Colini-BaldeschiLSZ14} and \cite{DBLP:journals/corr/BranzeiFMZ16} is that the former considered  bundle-pricing envy-freeness (pairwise envy-freeness), while the latter focused on the definition of item-pricing envy-freeness.
Recently, the item-pricing envy-free problem for general valuations in multi-unit markets without budgets is considered in \cite{DBLP:conf/ijcai/MonacoSZ15},
and a dynamic programming algorithm is provided in such setting.
Moreover, Chen and Deng \cite{chen2014envy} provide hardness results for unit-demand buyers in the multi-unit markets without budgets setting.

\textit{Mechanisms with budgets.} Dobzinski et al.~\cite{DBLP:conf/focs/DobzinskiLN08} provided a Pareto-optimal and incentive compatible mechanisms in multi-unit settings (i.e., multiple copies of the same items).  Similar results are obtained  in~\cite{FLSS11} where buyers have different preference sets. No incentive-compatible mechanism is fully efficient when buyers have budgets: only Pareto-optimality can be achieved. Moreover, it is crucial to assume that budgets and preferences are public knowledge for the design of incentive compatible, Pareto-optimal mechanisms.
If we insist on incentive compatible and Pareto-optimal mechanisms, no envy-free efficient allocation is possible.  Feldman et al.~\cite{DBLP:conf/sigecom/FeldmanFLS12} observed that existing mechanisms also sell identical items at different prices, e.g., the VCG mechanism for buyers with unlimited budget~\cite{clarke1971multipart,groves1973incentives,vickrey-61}, or the ascending price auction for buyers with budgets~\cite{DBLP:conf/focs/DobzinskiLN08,Nisan09}.

\textit{Ascending price auctions.} Ascending price auctions were used in FCC spectrum auctions and were initially studied in \cite{Ausubel04,ausubel-milgrom,Milgrom98putingauction}. Later, ascending price auctions were widely applied in the context of sponsored search auctions. Dobzinski et al.\cite{DBLP:conf/focs/DobzinskiLN08} designed a Pareto-optimal, incentive compatible ascending price auction for a multi-unit setting when buyers have budget constraints.
Numerous subsequent papers extended this setting, see \cite{DBLP:conf/icalp/BaldeschiHLS12,FLSS11,DBLP:conf/stoc/GoelML12}.

\section{Preliminaries}\label{sec:prelim}
An instance of the revenue-maximizing envy-free pricing problem in matching markets can be formally depicted by the tuple $\mathcal{A} = \langle I, J, \mathbf{S}, \mathbf{v}, \mathbf{b} \rangle$. There is a set of $|I|=n$ buyers and a set of $|J|=m$ different items in the market. Every buyer $i \in I$ is interested in a set of items $S_i \subseteq J$, which we refer as the preference-set of buyer $i$. Buyers equally value the items in their preference-sets. Specifically, every buyer $i$ has a valuation $v_i \in \mathbb{R}_{> 0}$ for each item $j \in S_i$ and has a valuation of zero for any item $j \notin S_i$. Every buyer $i$ has a budget $b_i \in \mathbb{R}_{\geq 0}$ that is the maximum payment she can afford. 

An algorithm computes an outcome  $\langle \mathbf{X}, \mathbf{p} \rangle$ for every possible instance $\mathcal{A}$, where $\mathbf{X} = \langle X_1, \ldots, X_n \rangle$ is the allocation vector, and $\mathbf{p}= \langle p_1,\ldots,p_n \rangle$ is the payment vector. That is, for each buyer $i$, $X_i \subseteq J$ is the set of items allocated to buyer $i$, and $p_i  \in \mathbb{R}_{\geq 0}$ is the payment charged to buyer $i$.
Moreover, we use $\bar{p}_{i}$ to denote the per-item-price paid by a buyer $i$, i.e., $\bar{p}_{i} = \frac{p_{i}}{|X_{i}|}$.
Given an allocation $X_i$ and a payment $p_i$, the utility of a buyer $i$ is defined as 

\[
u_i(X_i, p_i) = 
 \begin{cases} 
   v_i \cdot |X_i \cap S_i| - p_i & \text{if } b_{i} \geq p_{i}\\
   -\infty  & \text{if } b_{i} < p_{i}
  \end{cases}
\]

In addition, a feasible outcome $\langle \mathbf{X}, \mathbf{p} \rangle$ must satisfy the following constraints:

\begin{itemize}
	\item \textit{feasibility} (or \textit{supply constraint}): for any pair of buyers $i, i' \in I$,  $X_i \cap X_{i'} = \emptyset$;
	\item \textit{individual rationality}: for any buyer $i \in I$, $u_i(X_i, p_i) \geq 0$.
\end{itemize}

Furthermore, we incline toward our algorithm to produce envy-free outcomes. An outcome $\langle \mathbf{X}, \mathbf{p} \rangle$ is: 

\begin{itemize}
	\item \textit{envy-free}: if given an item-pricing vector $\mathbf{\rho} = \{\rho_{1}, \ldots, \rho_{m}\}$ such that the price for the item $j$ is $\rho_{j}$. Then
	$p_{i} = \sum_{{j \in X_{i}}} \rho_{j}$, and there is no bundle $X' \subseteq J$ such that $u_{i}(X',\sum_{{j \in X'}} \rho_{j}) > u_{i}(X_{i},p_{i})$. And,
	\item \textit{pairwise envy-free}: if given a set of proposed bundles $\mathbf{X}$ such that: every bundle $X_{i} \in \mathbf{X}$ has a corresponding price $p_{i} \in \mathbf{p}$, and every bundle $X' \not\in \mathbf{X}$ has price equal to $\infty$. Then for every buyer $i$ and every bundle $X_{j} \in \mathbf{X}$, buyer $i$ prefers her own bundle, i.e., $u_{i}(X_{i},p_{i}) \geq u_{i}(X_{j},p_{j})$.
\end{itemize}


Given outcome $\langle\mathbf{X},\mathbf{p}\rangle$, the revenue (i.e., the revenue of the algorithm on instance $\mathcal{A}$) is the sum of the payments of all buyers, i.e., $\mathcal{R}(\mathbf{X},\mathbf{p}) = \sum_{i \in I} p_i$. Our goal is to design an envy-free algorithm that approximates the optimal envy-free revenue for every possible instance $\mathcal{A}$. 

\subsection{Ascending price auction}
Our technique to design revenue-maximizing envy-free algorithms relies on the implementation of an ascending price auction. 
The standard implementation of an ascending price auction is as follows: (i) the price (initialized at zero) is raised until some condition is met (usually referred as \textit{selling condition}),
(ii) when the selling condition is met an appropriate \textit{selling procedure} is executed, and the auction goes back to (i) (if there is at least one buyer with positive demand).
The novelty of the ascending price auction described in this paper is about the selling procedure. We will give a detailed description of the selling procedures in Section \ref{sec:bgeqv} and Section \ref{sec:blessv}, but in either cases the selling procedure relies on the graph representation of the problems and their properties.

To facilitate our future presentation and analysis, here we introduce some necessary notations. For a price $p$, the \textit{demand} of buyer $i$ is defined as follows:


\begin{align}\label{eq:demand}
	D_i(p) &= \begin{cases} 
	\min\{\lfloor \frac{b_i}{p} \rfloor, |S_i|\}, & \mbox{if } p \leq v_i\\
	0, & \mbox{if } p > v_{i}
	\end{cases}
\end{align}

Intuitively, $D_i(p)$ is the number of the items that maximizes the utility of buyer $i$ if all items in $S_i$ are priced at $p$.

Given price $p$, we define three sets  of buyers $A^{p}$, $Q^{p}$, and $I^{p}$.
$A^{p}$ contains the buyers whose valuations are strictly greater than $p$ and having positive demands, i.e.,
$A^{p} = \{ i \in I | v_i > p \land D_i(p) > 0\}$.
$Q^{p}$ contains the buyers whose valuations are equal to $p$ and having positive demands, i.e., $Q^{p} = \{ i \in I | v_i = p \land D_i(p) > 0 \}$.
Finally, let $I^{p}$ be the union of $A^{p}$ and $Q^{p}$, i.e., $I^{p} =A^{p} \cup Q^{p}$.


\subsection{Graph representation}
Matching markets have a very intuitive bipartite graph representation. Consider a set of buyers $I' \subseteq I$ and a set of items $J' \subseteq J$ as two disjoint sets of nodes in a bipartite graph. Given a price $p$, there exists an edge between $i\in I'$ and $j \in J'$ if buyer $i$ demands item $j$ at price $p$, i.e., $(j \in S_i) \land (b_i \geq p) \land (p \leq v_i)$. More specifically, taking the notations established above, given a particular price $p$ and a subset of items  $J' \subseteq J$, we define a bipartite graph $G^p = (I^p \cup  J', E^p)$, where $E^p = \{(i,j)| i\in I^p, j \in J' \cap S_i\}$. Similarly, we define $\bar{G}^p = (A^p \cup J', \bar{E}^p)$ as the bipartite graph that only includes  buyers in $A^p$.  In this paper we refer to $G^p$ and $\bar{G}^p$ as \textit{preference-graphs}.  

Additionally, allocations in matching markets can be seen as matchings in preference-graphs, because an allocation essentially ``maps" buyers to a subset of items.  To formalize this idea, we introduce the concept of $B$-matching. 

\begin{definition}
Given a bipartite graph $G^p = (I^p \cup J', E^p)$, a $B$-matching $\mathcal{M}(G^p)$ is a 
sub-graph of $G^p$ such that every buyer is not matched to a number of items greater than her demand, i.e., $\forall i \in I^P$, $|\{ j \in J' : (i,j) \in \mathcal{M}(G^p)\}| \leq D_i(p)$, and every item is not matched to more than one buyer, i.e., 
$\forall j \in J^p$, $|\{i \in I^p : (i,j) \in \mathcal{M}(G^p)\}| \leq 1$.
\end{definition}
Similarly, we use $\mathcal{M}(\bar{G}^p) $ to denote the $B$-matching on graph $\bar{G}^p$. 
From now on, we will simply write matchings instead of $B$-matchings, and the notation $\mathcal{M}(G^p)$ will refer to a \textit{maximum matching} on $G^{p}$.
Given a matching, an allocation can be easily constructed. For example, if edge $(i,j)$ is in the matching, then item $j$ is allocated to buyer $i$. By the matching's definition, the allocations constructed from the matchings would satisfy the supply constraint and the budget constraint.

Finally, we introduce the concept of augmenting path. Given a preference-graph, an augmenting path starts with a buyer and ends with an unallocated item in the preference-graphs. We mainly use augmenting paths to produce envy-free outcomes. Intuitively, to achieve the envy-freeness, given a matching at price $p$, if a buyer obtains an item then every buyer in $A^p$, who is connected to that item by an augmenting path, will also obtain the items in their preference sets at price $p$.

\begin{definition}
Given a $B$-matching $\mathcal{M}(G^p)$ (resp. $\mathcal{M}(\bar{G}^p)$), an augmenting path $\pi$ between buyer $i$ and item $j$ is a path $\pi = \{i=y_1, z_1, y_2, z_2,\ldots, y_h, z_h=j\}$ such that the following conditions hold:
\begin{enumerate}
	\item $\forall k \in [1, \ldots, h]$, $y_k$ is a buyer. That is, $y_k \in I^p$(resp. $y_k \in A^p$).
	\item $\forall k \in [1, \ldots, h]$, $z_k$ is an item. That is, $z_k \in J'$. 
	\item $\forall k \in [1, \ldots, h]$, $z_k \in S_{y_k}$ and $(y_k, z_k) \notin \mathcal{M}(G^p)$ (resp. $(y_k, z_k) \notin \mathcal{M}(\bar{G}^p)$).
	\item $\forall k \in [1, \ldots, h-1]$, $(z_k, y_{k+1}) \in \mathcal{M}(G^p)$ (resp. $ (z_k, y_{k+1}) \in \mathcal{M}(\bar{G}^p)$).
\end{enumerate}
\end{definition}

In an augmenting path defined above, if item $z_h$ is allocated to a buyer $i \notin \{y_1,\ldots, y_h\}$ at price $p$, we will also allocate some item to  every buyer in $\{y_1,\ldots, y_h\}$. We use this technique to achieve the envy-freeness among buyers.

\section{An ascending price algorithm for $b_i \geq v_i$}\label{sec:bgeqv}
In this section we present an algorithm that obtains a $4$-approximation to the optimal bundle-price envy-free revenue when all  buyers have budgets that are equal to or greater than their valuations, i.e., $\forall i \in I, b_i \geq v_i$.
We remark that this is an item pricing algorithm, and it achieves a constant approximation with respect to the optimal bundle-price envy-free revenue.

The general idea is to implement an ascending price auction with selling conditions and selling procedures accurately designed to take advantage of this scenario.
The implementation of the ascending price auction is described in Algorithm \ref{alg:main}.
At the beginning the price is set to zero. Then the algorithm increases the price until a proper
selling condition is satisfied (line \ref{aucline:incprice}).
When the selling condition is matched, the algorithm executes a selling procedure and assigns items and payments to the involved buyers.
Notice that in Algorithm \ref{alg:main} there are two different selling procedures: \textsc{Compute-Allocation-I} described by Algorithm \ref{alg:VLalgorithm}, and \textsc{Compute-Allocation-II} described by Algorithm \ref{alg:NVLalgorithm}.
The selling condition and the two selling procedures will be described in the next subsections.

\subsection{Selling condition and critical prices}

Algorithm~\ref{alg:main} uses a selling condition to catch the correct price at which a selling procedure can be executed.
Intuitively, the notion of correct price relies  on the following abstraction: nothing can be sold if the buyers' cumulative demand is too high (the sum of the buyers' demands is greater than the available items).
But quantifying if buyers' cumulative demand is too high is not a trivial task for two main reasons: (i) in a matching markets setting the demand can be too high
on a particular set of items but too low on a different set, and (ii) the demand functions are not continuous, thus we can have a cumulative demand that is too high at $p$ and too low at $p+\epsilon$.

So, we want to detect each price $p$ that is borderline between a too high and a too low cumulative demand. These prices are called \textit{critical prices}.
To detect critical prices, we have to compute two maximum matchings for each price $p$: a maximum matching at price $p$ and a maximum matching at price $p+\epsilon$.
If the size of the maximum matching at $p$ is greater than the size of the maximum matching at $p+\epsilon$, then we know that on some set of items the cumulative demand
will be too low at any price higher than $p$. So, we have to check if something should be sold at $p$.


To be more formal, we define the set of critical prices as follows:

\begin{definition}
Given a price $p$, let $G^p = (I^p \cup J', E^p)$ where $J' \subseteq J$ is the set of unsold items. The price $p$ is \textit{critical} if $|\mathcal{M}(G^p)| > |\mathcal{M}(G^{p+\epsilon})|$, where $\mathcal{M}(G^p)$ and $\mathcal{M}(G^{p+\epsilon})$ are the maximum matchings in $G^p$ and $G^{p+\epsilon}$, respectively.
\end{definition}

Notice that the selling condition at line \ref{aucline:incprice} of Algorithm \ref{alg:main} is satisfied for each critical price.

\subsection{Detailed Description and Selling Procedures}

For each critical price $p$, Algorithm \ref{alg:main} performs the following actions:

\begin{itemize}
	\item The algorithm checks if price $p$ is critical because the demand of some buyer in $Q^p$ goes to zero when the price is increased above $p$ (line \ref{aucline:ifVL}). 
	In this case, Procedure \textsc{Compute-Allocation-I} computes an envy-free partial assignment, where every item is sold at $p$.
	
	\item The other reason of price $p$ being a critical price is that some buyer in $A^{p}$ cannot afford the same amount of items  at a slightly higher price (for budget limitations).
	In this case, Procedure \textsc{Compute-Allocation-II}  computes an envy-free
	partial  assignment, where every item is sold at $p+\epsilon$.
	
\end{itemize}

In either procedures the key element is the computation of an envy-free partial assignment.
In order to describe how an envy-free partial assignment is computed, we introduce some further notation.
Given a graph $G^{p}=(I^{p} \cup J', E^{p})$, a maximum matching $\mathcal{M}(G^{p})$ on $G^{p}$, and a subset of items $J' \subseteq J$, we denote with
$\bar{J} \subseteq J'$ the set of items that are not matched in $\mathcal{M}(G^{p})$, and with
 $N(\bar{J}) \subseteq I^{p}$ the set of buyers
that are connected with an augmenting path in $\mathcal{M}(G^{p})$ to some item in $\bar{J}$.
Now, an envy-free partial assignment is computed as follows:

\begin{enumerate}

\item a maximum matching $\mathcal{M}(G^{p})$ on the graph $G^{p}=(I^{p} \cup J, E^{p})$ is computed,

\item let $\bar{J} \subseteq J'$ be the subset of items that are not matched in $\mathcal{M}(G^{p})$,


\item the envy-free partial assignment $\textsc{m}$ is the subgraph of $\mathcal{M}(G^{p})$ that involves only buyers in $N(\bar{J})$.

\end{enumerate}

Restrict the allocation to the buyers in $N(\bar{J})$ is the key to obtain envy-freeness and a good revenue in a partial assignment.
Indeed, we claim that if a buyer $i$ is connected with an augmenting path to an item $j$, but that item $j$ remains unmatched in a maximum matching, then the buyer $i$ obtains as many items as she demands.
Intuitively, this is enough to extract a good revenue (allocating full demand we extract almost the whole budget), and to achieve an envy-free allocation (if a buyer receives her full demand, then she is happy).

The next paragraphs will discuss the details of the two procedures.
Procedure \textsc{Compute-Allocation-I} and Procedure \textsc{Compute-Allocation-II} implement the computation of envy-free partial assignment with different details,
but the two procedures are similar in spirit.

\begin{algorithm}
	\begin{algorithmic}[1]
		\Require $< I, J, \mathbf{S}, \mathbf{v}, \mathbf{b} >$
		\Ensure $\langle\mathbf{X},\mathbf{p} \rangle$
		\State{$p \gets 0$; /*$p$ is the uniform price for all items and it is dynamically increasing*/ }
		\State{$G^p = (I^p \cup J, E^p)$;  /*$G^p$ is the preference-graph for unsold items at price $p$ */} 
		\While{$|\mathcal{M}(G^p)| \neq \emptyset$} \label{aucline:activebuyers}
			
			\State{Increase $p$ until $|\mathcal{M}(G^p)| > |\mathcal{M}(G^{p+\epsilon})|$;} \label{aucline:incprice}
			
		\If{$|\mathcal{M}(G^p)    | > |\mathcal{M}(\bar{G}^p)|$} \label{aucline:ifVL}
		\State{$\textsc{M} \gets $\textsc{Compute-Allocation-I}$(I, J, \mathbf{S}, \mathbf{v}, \mathbf{b}, p)$;} \label{aucline:procedureVL}
		\For{each edge $(i,j) \in \textsc{M}$}
		\State{$J = J \setminus \{j\}$, $X_i = X_i \cup \{j\}$, $p_i = p_i + p$, $b_i = b_i - p$;}
		\EndFor
		\Else
		\State{$\textsc{M} \gets$\textsc{Compute-Allocation-II}$(I, J, \mathbf{S},\mathbf{v}, \mathbf{b},  p)$}; \label{aucline:procedureNVL}
		\For{each edge $(i,j) \in \textsc{M}$}
		\State{$J = J \setminus \{j\}$, $X_i = X_i \cup \{j\}$, $p_i = p_i +( p + \epsilon)$, $b_i = b_i - (p + \epsilon)$;}
		\EndFor
		\State{Remove all the items that are not demanded anymore.}
		\EndIf
		\EndWhile
	\end{algorithmic} \caption{An ascending price algorithm for matching markets  when $b_i \geq v_i$} \label{alg:main}
\end{algorithm}

\subsubsection{\textsc{Compute-Allocation-I}}
Algorithm \ref{alg:VLalgorithm} is executed when the price becomes equal to the valuation of some buyer, and it is not possible to sell the same amount of items at a slightly higher price. Notice that these buyers belong to the set $Q^{p}$.
Moreover, recall that, by definition of envy-freeness, every buyer in $Q^{p}$ cannot envy any other buyer that obtains items at price $p$.
Let us denote with $J^{p}_{Q}$ the set of available items in the preference sets of buyers in $Q^p$.
Thus, the general idea is to compute an envy-free partial assignment such that most of the items in $J^{p}_{Q}$ are assigned to the buyers in $A^p$.
In succession, the items that remain unmatched will be assigned to the buyers in $Q^{p}$.

To achieve this, the algorithm computes a maximum matching $\mathcal{M}(\bar{G}^p)$ with the minimal number of items in $J^p_Q$ matched (line \ref{aucline:computeVLmatching} in Algorithm \ref{alg:VLalgorithm}).
This can be done using the technique of weighted matching. 
In order to allocate the remaining items to buyers in $Q^p$ and preserve the envy-freeness of the outcome, the algorithm first allocates items to buyers in $A^p$ who have augmenting paths to any remaining item in  $\mathcal{M}(\bar{G}^p)$. 
Essentially, if an item is allocated to a buyer in $Q^p$ at price $p$, every buyer $i$ in $A^p$ who is connected to the item by an augmenting path obtains $\lfloor\frac{b_i}{p} \rfloor$ items and pays $p \cdot \lfloor\frac{b_i}{p} \rfloor$.  Thus, buyer $i$ will not envy the buyers in $Q^p$.

Every envy-free partial assignment $\textsc{M} \subseteq \mathcal{M}(\bar{G}^p)$ computed by procedure \textsc{Compute-Allocation-I} satisfies two important properties: (i) every buyer $i \in \textsc{M}$ receives an amount of items exactly equal to her demand, i.e., $|\{(i,j) \in \textsc{M}\}| = D_{i}(p)$ $\forall i \in \textsc{M}$, and (ii) every item that is not assigned at the end of procedure \textsc{Compute-Allocation-I} is requested by at least one buyer after procedure \textsc{Compute-Allocation-I}.

These properties are formally proved by the following lemmata.

\begin{lemma}\label{lem:vlallocation}
	Let $\mathcal{M}(\bar{G}^p)$ be a maximum matching on graph $\bar{G}^p$ 
	and $J_{Q}^{p} \neq \emptyset$, then each buyer $i \in N(\bar{J})$ is matched to $D_i(p)$ items in $\mathcal{M}(\bar{G}^p)$, where $\bar{J}$ is the set of items not matched in $\mathcal{M}(\bar{G}^p)$.
	
	
\end{lemma}

\begin{proof}
	Assume by contradiction that a buyer $i \in N(\bar{J})$
	gets less than $D_i(p)$ items. It implies that buyer $i$ is matched to less than $D_i(p)$ items in $\mathcal{M}(\bar{G}^p)$, that is, $|\{ (i,j) \in \mathcal{M}(\bar{G}^p) | j \in J \}| < D_i(p)$. 
	In addition, buyer  $i$ is connected to an item in $\bar{J}$ by an augmenting path. Therefore, one can assign one more item from $\bar{J}$ to buyers in $A^p$. It contradicts the fact that $\mathcal{M}(\bar{G}^p)$ is the maximum $B$-matching that minimizes the the number of items matched in $J^p_Q$.  \qed 
\end{proof}



\begin{lemma}\label{lem:vlunallocateditems}	
	In Procedure \textsc{Compute-Allocation-I},  all items in $\bar{J}$ can be allocated
	to the buyers in $Q^p$ at price $p$ per each.
\end{lemma}

\begin{proof}
$\bar{J}$ contains the items that are matched in $\mathcal{M}(G^p)$ but \textit{not} matched in $\mathcal{M}(\bar{G}^p)$. It implies that any item in $\bar{J}$ is in the preference-set of at least one buyer in $Q^p$. Given that $M(G^p)$ matches all the items, it guarantees that all unmatched items in $\bar{J}$ can be assigned to buyers in $Q^{p}$. 
\qed
\end{proof}

\begin{algorithm}
\begin{algorithmic}[1]
	\Procedure{Compute-Allocation-I}{$I, J, \mathbf{S}, \mathbf{v}, \mathbf{b}, p$}
	\label{procedure:VL}
	\State{Let $J_{Q}^p$ be the set of items in the preference-sets of  buyers in $Q^p$ at price $p$, i.e.
		$J_{Q}^p = \bigcup_{i \in Q^p} S_i$;} \label{aucline:VLpreferencesets}
	\State{Compute a  maximum $B$-matching $\mathcal{M}(\bar{G}^{p})$ with minimum number of items in $J_{Q}^p$ matched;} \label{aucline:computeVLmatching}
	\State{Let $\bar{J}$ be the set of items that are not matched in $\mathcal{M}(\bar{G}^{p});$} \label{aucline:setunsolditemsVL}
	\State{Let $N(\bar{J})$ be the set of buyers that are connected to an item in $\bar{J}$ with an augmenting path in $\mathcal{M}(\bar{G}^p)$;} \label{aucline:setneighbourbuyersVL}
	\State{Let $\textsc{M} = \{ (i,j) \in \mathcal{M}(\bar{G}^p) | i \in N(\bar{J}) \}$;} \label{aucline:partialmatchingVL}
	\State{Assign items in $\bar{J}$ to buyers in $Q^p$ such that the allocation satisfies the supply constraint and budget constraint;}
	\State{Include the assignment of $\bar{J}$ to $\textsc{M}$;}
	\State{{\bf return} $\textsc{M}$}
	\EndProcedure
\end{algorithmic}
\caption{\textsc{Compute-Allocation-I}}\label{alg:VLalgorithm}
\end{algorithm}

\subsubsection{\textsc{Compute-Allocation-II}}
Algorithm \ref{alg:NVLalgorithm} is executed when  it is not possible to sell the same amount of items at a slightly higher price, but no buyer in $Q^{p}$ is responsible for that.
Notice that in this case, the buyers that decrease their demands are in $A^{p}$. Namely, the buyers drop their demands because budgets and not valuations.
Thus, there are no buyers in $Q^{p}$ that are relevant for us. Consequently, it is enough to compute an envy-free partial assignment at price $p+\epsilon$.
Similarly to the previous case, to preserve the envy-freeness of the outcome, the algorithm allocates items to buyers in $A^p$ who have an augmenting path to an unmatched items in $\mathcal{M}(G^{p+\epsilon})$. Given the fact that all items are matched in $\mathcal{M}(G^p)$, if an items is not matched to any buyer in $\mathcal{M}(G^{p+\epsilon})$, then it implies that all buyers who have augmenting paths to this item must be fully matched in $\mathcal{M}(G^{p+\epsilon})$. In other words, those buyers do not have enough budgets to buy one more item. Otherwise, the size of $\mathcal{M}(G^{p+\epsilon})$ can be increased.
We prove this property in the following lemma (proof in Appendix \ref{appendix:nvlallocation}). 

\begin{algorithm}
\begin{algorithmic}[1]
	\Procedure{Compute-Allocation-II}{$I, J, \mathbf{S}, \mathbf{v}, \mathbf{b}, p$}
	\label{procedure:NVL}
	\State{Compute a  maximum $B$-matching $\mathcal{M}(G^{p + \epsilon})$;} \label{aucline:computeNVLmatching}
	\State{Let $\bar{J}^{p + \epsilon} $
		be the set of items that are not matched in $\mathcal{M}(G^{p + \epsilon})$;} \label{aucline:setunsolditemsNVL}
	\State{Let $N(\bar{J}^{p + \epsilon})$ be the set of buyers that are connected to an item $\bar{J}^{p + \epsilon}$ with an augmenting path;} \label{aucline:setneighbourbuyersNVL}
	\State{Let $\textsc{M} = \{ (i,j) \in \mathcal{M}(G^{p + \epsilon}) | i \in N(\bar{J}^{p + \epsilon}) \}$;} \label{aucline:partialmatchingNVL}
	\State{Remove items in $\bar{J}^{p + \epsilon}$ from $J$;}
	\State{{\bf return} $\textsc{M}$}
	\EndProcedure
\end{algorithmic}
\caption{\textsc{Compute-Allocation-II}}\label{alg:NVLalgorithm}
\end{algorithm}

\begin{lemma}\label{lem:nvlallocation}
	Let $\mathcal{M}(G^{p+\epsilon})$ be a maximum matching on graph $G^{p+\epsilon}$, and let $\bar{J}^{p+\epsilon} \neq \emptyset$ be the set of items not matched in the maximum matching $\mathcal{M}(G^{p+\epsilon})$. Then for each buyers $i \in N(\bar{J}^{p+\epsilon})$ the followings hold:
(i) $v_i > p$, and (ii) buyer $i$ is matched to $D_i(p+\epsilon)$ items in $\mathcal{M}(G^{p+\epsilon})$.
\end{lemma}


Notice that at the end of the procedure it is possible that some items remain unassigned but no buyer will demand them anymore.
This is a potential problem for the revenue, because many unassigned items can be translated in many not-extracted moneys.
But, the next lemma shows that the number of unassigned items is bounded (proof in Appendix \ref{appendix:boundunallocateditems}).

\begin{lemma}\label{lem:boundunallocateditems}
	Let $\mathcal{M}(G^{p+\epsilon})$ be a maximum matching on graph $G^{p+\epsilon}$, and let $\bar{J}^{p+\epsilon}$ be the set of items not matched in the maximum matching $\mathcal{M}(G^{p+\epsilon})$, then
	$|\bar{J}^{p+\epsilon} | \leq | N(\bar{J}^{p+\epsilon})|$. 
\end{lemma}


\subsection{Main result}

Finally, we are ready to prove that the outcome computed by Algorithm \ref{alg:main} is envy-free and achieves a $4$-approximation to the optimal bundle-price envy-free revenue.

We start with some auxiliary lemmata:

\begin{lemma}\label{lem:uniqueprice}
	Let $\langle\mathbf{X}, \mathbf{p}\rangle$ be the outcome obtained by Algorithm \ref{alg:main}. If $X_i \neq \emptyset$, the buyer $i$ obtains
	all the items in $X_i$ at a unique price-per-item $\bar{p}_i = \frac{p_i}{|X_i|}$.
\end{lemma}

The proof of Lemma \ref{lem:uniqueprice} is in Appendix~\ref{appendix:uniqueprice}. Then,
%
%
%

\begin{lemma} \label{lem:interestinprevprices}
	If buyer $i$ obtains $X_i$ at price-per-item $\bar{p}_i$, then all the items assigned at a price $p' < \bar{p}_{i}$ are not in her preference-set.
\end{lemma}
\begin{proof}
	Assume by contradiction that an item $j \in S_i$ is allocated at a price $p' < \bar{p}_{i}$. Then buyer $i$ has an augmenting path to $j$, and $i \in A^{p}$ (since her valuation is at least equal to $\bar{p}_{i}>p'$). Thus by Lemma~\ref{lem:nvlallocation}, we know that buyer $i$ would have obtained all her items at price $p'$.
	Then, it concludes that buyer $i$ is not interested in any item allocated at lower prices. 
	\qed
\end{proof}

\begin{lemma}\label{lem:unallocatedbuyers}
	If buyer $i$ does not obtain any item, i.e., $X_i = \emptyset$, then all items in $S_i$ are sold at a price greater than or equal to $v_i$.
\end{lemma}
\begin{proof}
	Observe that, since $i$ does not obtain any item, she will have a positive demand since until the price reaches her valuation.
Assume by contradiction that an item $j \in S_i$ has been sold at a price $p < v_i$.
Let $i'$ be the buyer that obtains $j$ at $p$. Then $i$ is connected with an augmenting path to $j$ at $p$.
But since $j \in S_i$, $D_i(p) > 0$, and $i \in A^{p}$, then by Lemma \ref{lem:nvlallocation} the buyer $i$ would have been allocated at $p$ as well. Contradiction.
	\qed
\end{proof}

Now we show that Algorithm \ref{alg:main} is envy-free.

\begin{theorem}
	The outcome $\langle \mathbf{X}, \mathbf{p}\rangle$ produced by Algorithm \ref{alg:main} is envy-free.
\end{theorem}

\begin{proof}
	First, by Lemma \ref{lem:unallocatedbuyers}, we know that the buyers that do not obtain any item do not
	envy anyone. Furthermore, by Lemma \ref{lem:uniqueprice}, we know that, for the rest of buyers, they obtain all items in $X_i$ at a unique per-item-price $\bar{p}_i = \frac{p_i}{|X_i|}$. The rest of the proof is divided into two cases.
	\begin{enumerate}
		\item $\bar{p}_i = v_i$: The buyer $i$ would not envy any buyer $j$ that gets her bundle $X_j$ at a price-per-item $\bar{p}_j \geq \bar{p}_i$. 
		Moreover, by Lemma \ref{lem:interestinprevprices}, buyer $i$ is not interested in any item allocated at a price $ p < \bar{p}_i$.
		This implies that buyer $i$ cannot envy any buyer $j$ who obtains the bundle at price-per-item $\bar{p}_j \leq \bar{p}_i$.
		\item $\bar{p}_i < v_i$:  Buyer $i$ obtains $X_i$ in either \textsc{Compute-Allocation-I} or \textsc{Compute-Allocation-II}. In both cases, by Lemma \ref{lem:vlallocation} and \ref{lem:nvlallocation}, buyer $i$ obtains $D_i(\bar{p}_i) = |X_i|$ items. Hence, buyer $i$ does not envy any buyer who obtains her bundle at price-per-item
		$p \geq \bar{p}_i$.
		Moreover, by Lemma \ref{lem:interestinprevprices}, buyer $i$ is not interested in any item allocated at a price $p < \bar{p}_i$.
		This implies that buyer $i$ cannot envy any buyer $j$ who obtains the bundle at price-per-item $\bar{p}_j \leq \bar{p}_i$.
	\end{enumerate}
	Thus we conclude that Algorithm~\ref{alg:main} is envy-free.
	\qed
\end{proof}

Now, we show that the outcome computed by Algorithm~\ref{alg:main} is a $4$-approximation to the optimal bundle-price envy-free revenue.

\begin{theorem}
	Algorithm~\ref{alg:main} achieves a $4-$approximation to the optimal envy-free revenue when all buyers have budgets that are at least their valuations.
\end{theorem}

\begin{proof}
	Let us recall some intuition of Algorithm~\ref{alg:main} first. For a given critical price $p$, Algorithm~\ref{alg:main}  performs either Procedure \textsc{Compute-Allocation-I} or \textsc{Compute-Allocation-II} to allocate items to buyers. In \textsc{Compute-Allocation-I}, some items are allocated to buyers that have valuations equal to the critical price. Note that some buyers whose valuations are greater than the critical price also obtain their items at the critical price if they are connected to those items. 
	On the other hand,  in \textsc{Compute-Allocation-II},  items are only allocated to buyers that have valuations greater than the critical price. Let us use  $I^{NVL}$ to denote the set of buyers who obtain items at prices lower than their valuations, and $I^{VL}$ to denote the set of buyers who obtain items at prices that are equal to their valuations. Finally, we use $I^{UN}$ to denote the set of buyers who obtain nothing at the end of Algorithm \ref{alg:main}. 
	In the remainder of this proof we will use the following additional notation. Let us denote with $\mathcal{R}^{\mathbf{OPT}}$ the optimal revenue, let $\mathcal{R}_{I^{NVL}}^{\mathbf{OPT}}$ be the fraction of the optimal revenue obtained from the buyers in $I^{NVL}$, and let $\mathcal{R}_{\mathcal{I}}^{\mathbf{OPT}}$ be the fraction of the optimal revenue obtained from the buyers in $\mathcal{I}$, where $\mathcal{I} = I^{VL} \cup I^{UN}$.
	Obviously, $\mathcal{R}^{\mathbf{OPT}} = \mathcal{R}_{I^{NVL}}^{\mathbf{OPT}} + \mathcal{R}_{\mathcal{I}}^{\mathbf{OPT}}$.
Moreover, let $J^{VL}$ be the set of items assigned to buyers in $I^{VL}$ at the end of Algorithm~\ref{alg:main}. Similarly, let $J^{NVL}$ be the set of items assigned to buyers in $I^{NVL}$ at the end of Algorithm~\ref{alg:main}. Given a set of items $J' \subseteq J$,  $\mathcal{R}^{\mathbf{}}(J')$ is the revenue that Algorithm~\ref{alg:main} obtains from items in $J'$.\footnote{All the new notation in this proof is only for the purpose of analysis.}

	The first step is to bound the optimal revenue that can be extracted to buyers in $I^{NVL}$. Let $\mathcal{R}_{I^{NVL}}^{\mathbf{OPT}}$ be the revenue from buyers in $I^{NVL}$ by any algorithm. 
	\begin{align*}
	\mathcal{R}_{I^{NVL}}^{\mathbf{OPT}} \leq \sum_{i \in I^{NVL}} b_i \overset{(a)}{\leq} \sum_{i \in I^{NVL}} \bar{p}_i
	(D_i(\bar{p}_i)+1) \overset{(b)}{\leq} 2 \cdot \sum_{i \in I^{NVL}} \bar{p}_i
	D_i(\bar{p}_i) \leq 2 \cdot \mathcal{R}(\mathbf{X}, \mathbf{p})
	\end{align*}
	where $\bar{p}_i = \frac{p_i}{|X_i|}$ is the price per item for buyer $i$. Inequality $(a)$ is guaranteed by Lemma~ \ref{lem:vlallocation} and ~\ref{lem:nvlallocation}, which shows that, at a given price, if a buyer whose valuation is greater than the price and is assigned to some items, then the buyers must be assigned to $D_i(\bar{p})$ items.  The inequality $(b)$ holds because for each buyer $i \in I^{NVL}$ we have that $D_i(\bar{p}_i) \geq 1$, since $b_i >  v_i$.
	
	The next step is to bound the optimal revenue that can be extracted from buyers in $\mathcal{I}$.
	Different from the previous case, as these buyers do not ``exhuast" their budgets, we need a different approach to bound their revenue. In particular, we bound the optimal revenue by considering the items in their preference sets. Let $P$ be the set of prices used by Algorithm~\ref{alg:main} i.e. $P = \{p \in \mathbf{p}\}$.
	
	\begin{align*}
	\mathcal{R}_{\mathcal{I}}^{\mathbf{OPT}} 
	\overset{(a)}{\leq} &\sum_{j \in \bigcup_{i\in \mathcal{I}} S_i} \max \{v_i | i \in \mathcal{I} \text{ and } j \in S_i\}\\
	\overset{(b)}{\leq} & \mathcal{R}^{\mathbf{}}(J^{VL}) + \mathcal{R}^{\mathbf{}}(J^{NVL} \cap \bigcup_{i\in \mathcal{I}} S_i)  
	+ \sum_{j \in \bigcup_{i\in \mathcal{I}} S_i \setminus \{J^{VL} \cup J^{NVL}\}} \max \{v_i | i \in \mathcal{I} \text{ and } j \in S_i\}\\
	\overset{(c)}{\leq} & \mathcal{R}^{\mathbf{}}( J^{VL} \cup J^{NVL}\cap \bigcup_{i\in \mathcal{I}} S_i) +  \sum_{p \in P} |N(J^{p+\epsilon})| \cdot p\\
	\overset{(d)}{\leq} & \mathcal{R}^{\mathbf{}}( J^{VL} \cup J^{NVL}\cap \bigcup_{i\in \mathcal{I}} S_i) + \mathcal{R}^{\mathbf{}}(J^{NVL}) \leq
	2 \cdot \mathcal{R}(\mathbf{X}, \mathbf{p})
	\end{align*}
	
	Now let us explain the bound above step by step. Inequality (a) comes from the fact that the optimal revenue is bounded by selling each item in their preference sets to the buyer with the highest valuation. For each item, there are three possibilities in our algorithm. The first case is that the item is sold to a buyer in $\mathcal{I}$. In this case, that buyer \textit{must} be the buyer with the highest valuation among all buyers who are interested in this item in $\mathcal{I}$. It is because if there are buyers who are also interested in this item and have higher valuation, those buyers must be assigned to other items at the same price (by the argument of augmenting path). Thus, these buyers are not belong to $\mathcal{I}$. The second case is that the item is sold to a buyer in $I^{NVL}$. In this case, the price of item is at least the maximum valuation among all buyers who are interested in this item in $I^{VL}$. The last case is the item remain unsold by the end of Algorithm~\ref{alg:main}. It happens in Procedure \textsc{Compute-Allocation-II}. For these items, we keep the same bound as before. By these arguments together, we obtain Inequality (b). 
	Next, Lemma~\ref{lem:vlunallocateditems} proved that after Procedure \textsc{Compute-Allocation-I} there are not new unallocated items, and 
	by Lemma~\ref{lem:boundunallocateditems} we know that the number of buyers allocated in Procedure \textsc{Compute-Allocation-II} is greater than the number of unsold items. Moreover, for each of those buyers, their payments are at least the highest valuation among all buyers who are interested in this item in $\mathcal{I}$. Otherwise, similar as before, buyers who are also interested in this item and have higher valuation must be assigned to other items at the same price (by the argument of augmenting path). Thus, these buyers are not in $\mathcal{I}$. Now we reach  Inequality (c). Inequality (d) is straightforward since each term is part of the revenue of Algorithm~\ref{alg:main}.
	
Putting all pieces we got so far together, we have
	
	\begin{align*}
	\mathcal{R}^{\mathbf{OPT}} = \mathcal{R}_{I^{NVL}}^{\mathbf{OPT}} + \mathcal{R}_{\mathcal{I}}^{\mathbf{OPT}} 
	\leq 4 \cdot \mathcal{R}(\mathbf{X}, \mathbf{p})
	\end{align*}
	\qed
\end{proof}

\section{An ascending price algorithm for $b_i < v_i$}\label{sec:blessv}
In this section we present an algorithm that obtains a $4$-approximation to the optimal bundle-price envy-free revenue in matching markets when all  buyers have budgets less than their valuations, i.e., $\forall i \in I, b_i < v_i$.
The main difference with respect to the algorithm presented in Section \ref{sec:bgeqv} is about the adopted pricing scheme.
Due to a separation example showed in Feldman et al.~\cite{DBLP:conf/sigecom/FeldmanFLS12}, we know that the optimal bundle-price envy-free revenue cannot be approximated within $O(\frac{1}{m})$ using an item pricing scheme.
Thus the algorithm presented in this section will embody a bundle pricing scheme, and it will produce a pairwise envy-free outcome.\footnote{Precisely, following the pricing scheme definitions introduced in \cite{DBLP:conf/sigecom/FeldmanFLS12}, the adopted pricing scheme is a $(h,p)$-proportional pricing scheme. It means that we impose an upper bound on the size of the bundles that the buyers may request.}
Before the description of the algorithm, we overload some notations. Given price $p$, let 
$A^{p}$ contain the buyers whose budget are strictly greater than $p$ and having positive demands, i.e.,
$A^{p} = \{ i \in I | b_i > p \land D_i(p) > 0\}$. Let 
$Q^{p}$ contain the buyers whose budgets are equal to $p$ and having positive demands, i.e., $Q^{p} = \{ i \in I | b_i = p \land D_i(p) > 0 \}$.
Finally, let $I^{p}$ be the union of $A^{p}$ and $Q^{p}$, i.e., $I^{p} =A^{p} \cup Q^{p}$.

\begin{algorithm}
	\begin{algorithmic}[1]
		\Require $< I, J, \mathbf{S}, \mathbf{v}, \mathbf{b} >$
		\Ensure $\langle\mathbf{X},\mathbf{p} \rangle$
		\State{$p \gets 0$; /*$p$ is the price per-item and it is dynamically increasing*/ }
		\State{$G^p = (I^p \cup J, E^p)$;  /*$G^p$ is the preference-graph for unsold items at price $p$ */} 
		\While{$|\mathcal{M}(G^p)| \neq \emptyset$} 
			\State{Increase $p$ until $|\mathcal{M}(G^p)| > |\mathcal{M}(G^{p+\epsilon})|$;}
			\State{Partition $(N(\bar{J}^p) \cup Q^p) $ into sets of buyers $Y_1, \ldots, Y_k$, where for any $Y_t$ and $Y_{t'}$ there do not exist buyers $y \in Y_t$ and $y' \in Y_{t'} $ such that $S_y \cap S_{y'} \neq \emptyset$;}
			\For{each $Y_t$}
		\If{$\sum_{i \in (Y_t \cap Q^p)\setminus \hat{I} } b_i < \sum_{i \in Y_t \cap N(\bar{J}^p)   }b_i$} 
		\State{$\textsc{M} \gets $\textsc{Compute-Allocation-II}$(Y_t , J, \mathbf{S}, \mathbf{v}, \mathbf{b}, p+\epsilon)$;} 
		\For{each edge $(i,j) \in \textsc{M}$}
		\State{$\hat{I} = \hat{I} \cup \{i\}$, $J = J \setminus \{j\}$, $X_i = X_i \cup \{j\}$, $p_i = p_i + p$, $b_i = b_i - p$;}
		\EndFor
		\Else
				\If{$\exists$ a matching $\textsc{M}$ s.t. every buyer in $Y_t \setminus \hat{I}$ is matched to one item}
				\State{Compute such a matching $\textsc{M}$;}
				\Else
						\State{$\textsc{M} \gets $\textsc{Compute-Allocation-II}$(Y_t, J,\mathbf{S}, \mathbf{v}, \mathbf{b}, p+\epsilon)$;} 
				\EndIf
				\For{each edge $(i,j) \in \textsc{M}$}
				\State{$\hat{I} = \hat{I} \cup \{i\}$, $J = J \setminus \{j\}$, $X_i = X_i \cup \{j\}$, $p_i = p_i + p$, $b_i = b_i - p$;}
				\EndFor
		\EndIf
		\EndFor
		\EndWhile
	\end{algorithmic} \caption{An ascending price algorithm for matching markets when $v_i > b_i$} \label{alg:main2}
\end{algorithm}

\subsection{Detailed Description}
The algorithm, which is referred to as Algorithm~\ref{alg:main2}, shares the similar spirit as Algorithm~\ref{alg:main} but possesses some tweaks on the actions performed at critical prices. 
The algorithm starts with an initial price $p=0$ for all items and keeps increasing the price for all items until the price becomes a critical price. 
The reason of a price being a critical price is a little different from the previous case.
Since $v_i > b_i$ for all buyers, a price becomes a critical price when it is equal to the budget of some buyer, or buyers cannot afford the same amount of items at higher prices. At each critical price, the algorithm divides buyers in $(N(\bar{J}^p) \cup Q^p)$ into different partition (a partition is denoted by $Y_{t}$). One property of these partitions is that no buyers from different partitions have a common item in their preference sets. It allows us to focus on each partition separately. Then the following actions are performed on each partition.  

\begin{itemize}
	\item The algorithm compares the remaining budgets between buyers in $(Y_t \cap Q^{p} )  \setminus \hat{I}$ and buyers in $Y_t \cap N(\bar{J}^{p})$ where $\hat{I}$ is the set of buyers who have obtained items at previous critical prices and still have budgets to demand more items. Recall that $\bar{J}^{p}$ is the set of items that are not matched in $\mathcal{M}(\bar{G}^{p})$ and $N(\bar{J}^{p})$ is the set of buyers that are connected to an item in $\bar{J}^{p}$ in  $\mathcal{M}(\bar{G}^{p})$. If the sum of the budgets of buyers in $(Y_t \cap Q^{p} )  \setminus \hat{I}$ is relatively small, then the algorithm  ``ignores" them (i.e. does not allocate them any item) but allocates items to buyers in $Y_t \cap N(J^{p})$ at $p+\epsilon$ each. It can be achieved by the same as Procedure \textsc{Compute-Allocation-II}. It would extract at least half of  the budgets of buyers in $Y_t \cap N(J^{p})$, which in turn is a good approximation to the optimal revenue from all buyers in $Y_t \setminus \hat{I}$.
	\item On the other case, when the sum of the budgets of buyers in  $(Y_t \cap Q^{p} )  \setminus \hat{I}$ is relatively large, the algorithm checks if it is possible to give one item to every buyer in $Y_t \setminus \hat{I}$. If yes, the algorithm allocates one item to each of them. By doing so, the algorithm extracts all the budgets of buyers in $(Y_t \cap Q^{p}) \setminus \hat{I} $ since their budgets are equal to the price. It gives us a good approximation to the optimal revenue extracted from buyers in $Y_t$. Otherwise, the algorithm ``ignores" buyers in  $(Y_t \cap Q^{p})\setminus \hat{I}$ and allocate items to buyers in $Y_t \cap N(\bar{J}^p)$ at price $p+\epsilon$ each. We show that it does not hurt the revenue since the optimal envy-free algorithm cannot extract any revenue from those buyers either.  
\end{itemize}	

\subsection{Main result}
Our main result is the following. Due to the space limit, the proofs of envy-freeness and the revenue guarantee  of Algorithm~\ref{alg:main2} are in Appendix~\ref{sec:envy-freeA2} and~\ref{sec:approximationA2}. 

\begin{theorem}
	Algorithm~\ref{alg:main2} is pairwise envy-free and achieves a $4-$approximation to the optimal revenue in envy-free outcomes when all buyers have budgets less than their valuations.
\end{theorem}

\bibliographystyle{plain}
\bibliography{ref}

\begin{thebibliography}{10}

\bibitem{Ausubel04}
Lawrence~M. Ausubel.
\newblock An efficient ascending-bid auction for multiple objects.
\newblock {\em American Economic Review}, 94(5):1452--1475, December 2004.

\bibitem{ausubel-milgrom}
Lawrence~M. Ausubel and Paul~R. Milgrom.
\newblock Ascending auctions with package bidding.
\newblock {\em Frontiers of Theoretical Economics}, 1:1019--1019, 2002.

\bibitem{BalcanBM08}
Maria-Florina Balcan, Avrim Blum, and Yishay Mansour.
\newblock Item pricing for revenue maximization.
\newblock In {\em Proceedings of the 9th ACM Conference on Electronic
  Commerce}, pages 50--59, 2008.

\bibitem{DBLP:journals/corr/BranzeiFMZ16}
Simina Br{\^{a}}nzei, Aris Filos{-}Ratsikas, Peter~Bro Miltersen, and Yulong
  Zeng.
\newblock Envy-free pricing in multi-unit markets.
\newblock {\em CoRR}, abs/1602.08719, 2016.

\bibitem{Briest:2006:SUS:1109557.1109678}
Patrick Briest and Piotr Krysta.
\newblock Single-minded unlimited supply pricing on sparse instances.
\newblock In {\em Proceedings of the Seventeenth Annual ACM-SIAM Symposium on
  Discrete Algorithm}, pages 1093--1102, 2006.

\bibitem{DBLP:conf/approx/ChalermsookCKK12}
Parinya Chalermsook, Julia Chuzhoy, Sampath Kannan, and Sanjeev Khanna.
\newblock Improved hardness results for profit maximization pricing problems
  with unlimited supply.
\newblock In {\em APPROX-RANDOM}, pages 73--84, 2012.

\bibitem{chen2014envy}
Ning Chen and Xiaotie Deng.
\newblock Envy-free pricing in multi-item markets.
\newblock {\em ACM Transactions on Algorithms (TALG)}, 10(2):7, 2014.

\bibitem{CS08}
Maurice Cheung and Chaitanya Swamy.
\newblock Approximation algorithms for single-minded envy-free
  profit-maximization problems with limited supply.
\newblock In {\em The 49th Annual Symposium on Foundations of Computer
  Science}, pages 35--44, 2008.

\bibitem{clarke1971multipart}
Edward~H Clarke.
\newblock Multipart pricing of public goods.
\newblock {\em Public choice}, 11(1):17--33, 1971.

\bibitem{DBLP:conf/icalp/BaldeschiHLS12}
Riccardo Colini{-}Baldeschi, Monika Henzinger, Stefano Leonardi, and Martin
  Starnberger.
\newblock On multiple keyword sponsored search auctions with budgets.
\newblock In {\em ICALP}, pages 1--12, 2012.

\bibitem{DBLP:conf/wine/Colini-BaldeschiLSZ14}
Riccardo Colini{-}Baldeschi, Stefano Leonardi, Piotr Sankowski, and Qiang
  Zhang.
\newblock Revenue maximizing envy-free fixed-price auctions with budgets.
\newblock In {\em the proceedings of the 10th International Conference on Web
  and Internet Economics (WINE)}, pages 233--246, 2014.

\bibitem{DBLP:conf/soda/DemaineHFS06}
Erik~D. Demaine, Mohammad~Taghi Hajiaghayi, Uriel Feige, and Mohammad~R.
  Salavatipour.
\newblock Combination can be hard: approximability of the unique coverage
  problem.
\newblock In {\em SODA}, pages 162--171, 2006.

\bibitem{DBLP:conf/focs/DobzinskiLN08}
Shahar Dobzinski, Ron Lavi, and Noam Nisan.
\newblock Multi-unit auctions with budget limits.
\newblock In {\em the Proceedings of the 49th Annual Symposium on Foundations
  of Computer Science}, pages 260--269, 2008.

\bibitem{DBLP:conf/sigecom/FeldmanFLS12}
Michal Feldman, Amos Fiat, Stefano Leonardi, and Piotr Sankowski.
\newblock Revenue maximizing envy-free multi-unit auctions with budgets.
\newblock In {\em in proceedings of the 13th ACM Conference on Electronic
  Commerce (EC)}, pages 532--549, 2012.

\bibitem{FLSS11}
Amos Fiat, Stefano Leonardi, Jared Saia, and Piotr Sankowski.
\newblock Single valued combinatorial auctions with budgets.
\newblock In {\em in proceedings of the 12th ACM Conference on Electronic
  Commerce (EC)}, pages 223--232, 2011.

\bibitem{Foley}
D.~Foley.
\newblock Resource allocation and the public sector.
\newblock {\em Yale Economic Essays}, 7:45--98, 1967.

\bibitem{DBLP:conf/stoc/GoelML12}
Gagan Goel, Vahab~S. Mirrokni, and Renato~Paes Leme.
\newblock Polyhedral clinching auctions and the adwords polytope.
\newblock In {\em Proceedings of the 44th Symposium on Theory of Computing
  Conference}, pages 107--122, 2012.

\bibitem{groves1973incentives}
Theodore Groves.
\newblock Incentives in teams.
\newblock {\em Econometrica: Journal of the Econometric Society}, pages
  617--631, 1973.

\bibitem{GHKKKM05}
Venkatesan Guruswami, Jason~D. Hartline, Anna~R. Karlin, David Kempe, Claire
  Kenyon, and Frank McSherry.
\newblock On profit-maximizing envy-free pricing.
\newblock In {\em Proceedings of the Sixteenth Annual ACM-SIAM Symposium on
  Discrete Algorithms}, pages 1164--1173, 2005.

\bibitem{Milgrom98putingauction}
Paul Milgrom.
\newblock Putting auction theory to work: The simultaneous ascending auction.
\newblock {\em Journal of Political Economy}, 108:245--272, 1998.

\bibitem{DBLP:conf/ijcai/MonacoSZ15}
Gianpiero Monaco, Piotr Sankowski, and Qiang Zhang.
\newblock Revenue maximization envy-free pricing for homogeneous resources.
\newblock In {\em Proceedings of the Twenty-Fourth International Joint
  Conference on Artificial Intelligence, {IJCAI} 2015, Buenos Aires, Argentina,
  July 25-31, 2015}, pages 90--96, 2015.

\bibitem{Nisan09}
Noam Nisan, Jason Bayer, Deepak Chandra, Tal Franji, Robert Gardner, Yossi
  Matias, Neil Rhodes, Misha Seltzer, Danny Tom, Hal~R. Varian, and Dan
  Zigmond.
\newblock Google's auction for tv ads.
\newblock In {\em ICALP}, pages 309--327, 2009.

\bibitem{Varian}
Hal~R. Varian.
\newblock Equity, envy, and efficiency.
\newblock {\em Journal of Economic Theory}, 9:63--91, 1974.

\bibitem{vickrey-61}
William Vickrey.
\newblock Counterspeculation, auctions, and competitive sealed tenders.
\newblock {\em The Journal of finance}, 16(1):8--37, 1961.

\end{thebibliography}

\newpage

\appendix

\section{Proof of Lemma \ref{lem:nvlallocation}}\label{appendix:nvlallocation}

\begin{proof}
	First, we prove $(i)$. Since $|\mathcal{M}(G^p)| > |\mathcal{M}(G^{p+\epsilon})|$, but $|\mathcal{M}(G^p)| = |\mathcal{M}(\bar{G}^{p})|$, then the valuations of all buyers in the graph $G^p$ are greater than $p+\epsilon$.
Thus for any  buyers $i \in N(\bar{J}^{p+\epsilon})$, it holds $v_i > p$. Next we prove $(ii)$. Assume by contradiction that a buyer $i \in N(\bar{J}^{p+\epsilon})$ gets allocated $|\{ (i,j) \in \mathcal{M}(G^{p+\epsilon}) \}|  < D_i(p+\epsilon)$ items. 
Since $i$ can afford at least one more item at $p+\epsilon$ and $i$ is connected to an item $j \in \bar{J}^{p+\epsilon}$ by an augmenting path $\pi$ (it is implied by the fact all items can be sold at price $p$), it implies that one can match one more item in $G^{p+\epsilon}$. It contradicts the fact that $\mathcal{M}(G^{p+\epsilon})$ is a maximum matching. 
\qed
\end{proof}

\section{Proof of Lemma \ref{lem:boundunallocateditems}}\label{appendix:boundunallocateditems}

\begin{proof}
Notice the maximum $B$-matching on $G^p$ matches all items in $G^p$.
Notice also that all the buyers in Algorithm \ref{alg:VLalgorithm} belong to $A^{p}$, thus every buyer decreases her demand at most by $1$.
So, since $N(\bar{J}^{p+\epsilon}) \subseteq A^{p}$ is the set of buyers that are interested in the items $\bar{J}^{p+\epsilon}$, we have that $|\bar{J}^{p+\epsilon} | \leq | N(\bar{J}^{p+\epsilon})|$.
\qed
\end{proof}

\section{Proof of Lemma \ref{lem:uniqueprice}}\label{appendix:uniqueprice}
\begin{proof}
	Assume by contradiction that buyer $i$ obtains the items in $X_i$ at different prices. We denote this set of prices by $P$.
	Let $\bar{p} = \min \{ p | p \in P\}$. Let $j \in X_i$ be an item assigned to $i$ at price $\bar{p}$.
	Note that Algorithm \ref{alg:main} assigns items to buyers in either Procedure \textsc{Compute-Allocation-I} or \textsc{Compute-Allocation-II}. We prove that if the buyer $i$ obtains an item at $\bar{p}$, then she has to obtain all the other items at $\bar{p}$ too.
	
	\begin{itemize}
		
		\item Suppose that the item $j$ is assigned in \textsc{Compute-Allocation-I}. It means that either $(a)$
		$i \in Q^p$ or $(b)$ $i \in A^p$ and $i$ is linked to some item in $\bar{J}$ by an 
		augmenting path.
		In case $(a)$, $v_i = \bar{p}$, thus buyer $i$ cannot obtain any item at price $p > \bar{p}$. 
		In case $(b)$, by Lemma \ref{lem:vlallocation}, buyer $i$ obtains $D_i(\bar{p})$ items at price $\bar{p}$ each,
		thus at any price $p > \bar{p}$ she cannot afford any item.
		
		\item Suppose that the item $j$ is assigned in \textsc{Compute-Allocation-II}.
		It means that $i\in A^P$ and is linked to some item in $\bar{J}^{p+\epsilon}$ by an augmenting path.
		By Lemma \ref{lem:nvlallocation}, buyer $i$ obtains $D_i(\bar{p})$ items at price $\bar{p}$ each. 
		Thus at any price $p > \bar{p}$, she cannot afford any item.
	\end{itemize}
	\qed
\end{proof}

\section{Envy-freeness of Algorithm~\ref{alg:main2}}
\label{sec:envy-freeA2}
In this section, we show that Algorithm~\ref{alg:main2} is pairwise envy-free. First, we show that if a buyer does not obtain any item, then every item in her preference-set is sold to another buyer at price greater than her budget. 

\begin{lemma}\label{lem:unallocatedbidders_budget}
	If buyer $i$ does not obtain any item, i.e., $X_i = \emptyset$, then all items in $S_i$ are sold at a price greater than $b_i$.
\end{lemma}
\begin{proof}
	Assume that an item $j \in S_i$ is sold at a price $p \leq b_i$.
	Let $i'$ be the buyer that obtains $j$ at $p$. Since $i'$ obtains $j$ at $p$, it implies that $i'$ has an augmenting path to an unmatched item $j' \in \bar{J}^p$ in $\mathcal{M}(G^p)$ at $p$. 
	Thus also $i$ has an augmenting path to $j'$. Hence, $i$ will be allocated to some item as $p$. It reaches a contradiction.
	\qed
\end{proof}

Second, we show that if an item $j \in S_i$ is sold at price $p$, then buyer $i$ whose budget is great than $p$ will obtain some other item(s) in $S_i$ at price-per-item $p$ as well. 

\begin{lemma}\label{lem:gettingsomething}
	If an item $j \in S_i$ is sold at $p$ and buyer $i$ has positive demand at $p$, then $i$ obtains some items(s) in $S_i$ at $p$.
\end{lemma}

\begin{proof}
	Assume by contradiction that $j\in S_i$ is sold to buyer $i'$ at $p$ and buyer $i$ does not obtain anything if she has a positive demand.
	Since $i'$ gets item $j$ it means that $i$ is connected through an augmenting path to an item $j' \in \bar{J}^p$.
	As $j \in S_i$, it also implies that $i$ has an augmenting path to $j'$ and will obtain some items in $S_i$ in Algorithm~\ref{alg:main2}. This reaches a contradiction.\qed
\end{proof}

Finally, we show that when buyer $i$ obtains items in $Z_i \subseteq X_i$ at price-per-item $p$,  buyer $i$ does not envy the bundles allocated to other buyers at the same price. 

\begin{lemma} \label{lem:maximumbundle}
	If buyer $i$ obtains a bundle $Z_i$ at price $p$ then $(v_i - p) \cdot |Z_i| \geq (v_i - p) \cdot |Z_j|$ for all $j \neq i$, where $Z_j$ is the set of items obtained by buyer $j$.
\end{lemma}
\begin{proof}
	Buyer $i$ obtains items in two situations: $(a)$ when all buyers in $Y_{t} \cap N(\bar{J}^P)$ obtain the number of items that are equal to their demands, or $(b)$ when all buyers in 
	$Y_t$ obtain exactly $1$ item. It is clear that in both situations buyer $i$ gets the best allocation among all buyers. 
	\qed
\end{proof}

Now we prove that Algorithm~\ref{alg:main2} is pairwise envy-free. By Lemma \ref{lem:unallocatedbidders_budget}, we know that the buyers who do not obtain any item in Algorithm~\ref{alg:main2} will not envy anyone since all items in their preference sets are sold at prices higher than their budgets. By Lemma \ref{lem:gettingsomething}, if an item $j \in S_i$ is sold at price $p$ and $D_i(p)>0$ then buyer $i$ gets at least one item too. Moreover, by Lemma \ref{lem:maximumbundle} we know that $i$ obtains the bundles that maximizes her utility among all the bundles allocated at $p$. This concludes that the outcome in Algorithm~\ref{alg:main2} is pairwise envy-free.

\section{Revenue guarantee of Algorithm~\ref{alg:main2}}
\label{sec:approximationA2}
In this section, we analyze the approximation ratio of Algorithm~\ref{alg:main2}. The main result is the following.

\begin{theorem}\label{thm:approximationA2}
	Algorithm \ref{alg:main2} obtains a $4$-approximation with respect to the optimal revenue,
	i.e. $\mathcal{R}(\mathbf{X}, \mathbf{p}) \geq \frac{\mathcal{R}^{\mathbf{OPT}}}{4}$.
\end{theorem}

\begin{proof}
	Given $I' \subseteq I$, let $\mathcal{R}_{I'}^{\mathbf{OPT}}$ denote the revenue extract from buyers in $I'$ in the optimal outcome.
	Similarly, $\mathcal{R}_{I'}(\mathbf{X},\mathbf{p})$ denote the revenue extracted by our algorithm in allocation $(\mathbf{X},\mathbf{p})$ from buyers in $I'$.
	Now, we partition $I$ in three subsets of buyers according their allocations in Algorithm~\ref{alg:main2}:
	
	\begin{itemize}
		
		\item $I_1$: 
		for each critical price $p$ and each set of buyers $Y_{t}$ such that
		$\sum_{i \in (Y_{t} \cap Q^p) \setminus \hat{I}} b_i < \sum_{i \in Y_{t} \cap N(\bar{J}^p)} b_i$, all buyers in
		$Y_{t} \cap (Q^p \cup N(\bar{J}^p) )\setminus \hat{I}$ are added to $I_1$.
		
		\item $I_2$: 
		for each critical price $p$ and each set of buyers $Y_{t}$ such that
		$\sum_{i \in (Y_{t} \cap Q^p) \setminus \hat{I}} b_i \geq \sum_{i \in Y_{t} \cap N(\bar{J}^p)} b_i$, all sets of buyers
		$Y_{t} \cap (Q^p \cup N(\bar{J}^p)) \setminus \hat{I}$ to whom is possible to allocate $1$ item each are added to $I_2$.
		
		\item $I_3$: 
		for each critical price $p$ and each set of buyers $Y_{t}$ such that
		$\sum_{i \in (Y_{t} \cap Q^p) \setminus \hat{I}} b_i \geq \sum_{i \in Y_{t} \cap N(\bar{J}^p)} b_i$, all sets of buyers
		$Y_{t} \cap (Q^p \cup N(\bar{J}^p)) \setminus \hat{I}$ to whom is not possible to allocate $1$ item each are added to $I_3$.
		
	\end{itemize}
	
	 Notice that each time we add a set of buyers in $I_{1}$, $I_{2}$, or $I_{3}$, we do not consider buyers in $\hat{I}$.
	 Thus to these sets are added only the buyers in $(Q^p \cup N(\bar{J}^p) )$ that do not allocate anything before.
	
	So, $I = I_1 \cup I_2 \cup I_3$ and $I_i \cap I_j = \emptyset$ for $i,j \in \{1,2,3\}, i \neq j$.
	
	Now we need to prove three auxiliary lemmata.
	
	\begin{lemma}\label{lem:I1}
		$\mathcal{R}_{I_{1}}^{\mathbf{OPT}} \leq 2 \mathcal{R}_{I_{1}}(\mathbf{X},\mathbf{p})+ 2 \mathcal{R}(\mathbf{X},\mathbf{p})$
	\end{lemma}
	
	\begin{proof}
		
		Denote with $I^{>0}_1 = \{i \in I_1 | X_i \neq \emptyset \}$ and $I^{=0}_1 = \{i \in I_1 | X_i = \emptyset \}$.
		\begin{align*}
		\mathcal{R}_{I_{1}}^{\mathbf{OPT}} \leq \sum_{i \in I_1} b_i =  \sum_{i \in I_1^{>0}} b_i + \sum_{i \in I_1^{=0}} b_i 
		\leq 2 \mathcal{R}_{I_{1}}(\mathbf{X},\mathbf{p})+ 2 \mathcal{R}(\mathbf{X},\mathbf{p})
		\end{align*}
		
		where the last inequality is because to each buyer $i \in I^{>0}_1$, Algorithm \ref{alg:main2} sells $\lfloor \frac{b_i}{p}\rfloor$ items, where $p_i$ is the price
		in which $i$ is inserted in $I_1$. Thus
		\begin{align*}
		\sum_{i \in I_1^{>0}} b_i
		\leq  \sum_{i \in I_1^{>0}} \bigg(\bigg\lfloor \frac{b_i}{p_i}\bigg\rfloor + 1\bigg) \cdot p_i
		\leq  \sum_{i \in I_1^{>0}} 2 \bigg\lfloor \frac{b_i}{p_i}\bigg\rfloor \cdot p_i = 2 \mathcal{R}_{I_{1}}(\mathbf{X},\mathbf{p})
		\end{align*}
		since $\lfloor \frac{b_i}{p}\rfloor \geq 1$ for all $i \in I^{>0}_1$.
		And $\sum_{i \in I_1^{=0}} b_i \leq 2 \mathcal{R}(\mathbf{X},\mathbf{p})$ since for all buyers $I^{=0}_1$ exist a set of buyers $I' \subseteq I$ such that
		$\sum_{i \in I^{=0}_1} b_i \leq \sum_{i \in I'} b_i$.
		\qed
	\end{proof}
	
	\begin{lemma}\label{lem:I2}
		$\mathcal{R}_{I_{2}}^{\mathbf{OPT}} \leq 2 \mathcal{R}_{I_{2}}(\mathbf{X},\mathbf{p})$
	\end{lemma}
	
	\begin{proof}
		
		Denote with $I^{=1}_2$ buyers that were budget-limited when added to $I_2$ and let $I^{>1}_2$ be the set of buyers that were not budget-limited when added to $I_2$.
		
		\begin{align*}
		\mathcal{R}_{I_{2}}^{\mathbf{OPT}} \leq \sum_{i \in I_2} b_i = \sum_{i \in I^{=1}_2} b_i + \sum_{i \in I^{>1}_2} b_i \leq 2\sum_{i \in I^{=1}_2} b_i =
		2 \mathcal{R}_{I_{2}}(\mathbf{X},\mathbf{p})
		\end{align*}
		
		The second inequality is because $\sum_{i \in I_2^{=1}} b_i > \sum_{i \in I_2^{>1}} b_i$ by definition of $I_2$.\qed
		
	\end{proof}
	
	\begin{lemma}\label{lem:I3}
		$\mathcal{R}_{I_{1}}^{\mathbf{OPT}} \leq 2 \mathcal{R}_{I_{3}}(\mathbf{X},\mathbf{p})$
	\end{lemma}
	
	\begin{proof}
		
		Denote with $I^{>0}_3 = \{i \in I_3 | X_i \neq \emptyset \}$ and $I^{=0}_3 = \{i \in I_3 | X_i = \emptyset \}$.
		First start with the following lemma.
		
		\begin{lemma}\label{lem:optdoesnotallocatetoo}
			
			All buyers in $I^{=0}_3$ do not receive items also in the optimal envy-free solution.
			
		\end{lemma}
		
		\begin{proof}
			
			Let $\mathcal{P}$ be the set of prices used for buyers in $I_3$.
			Thus, we can partition the set $I^{=0}_3 = \bigcup_{p \in \mathcal{P}} I^{=0}_{3,p}$.
			
			Now, for all each $p \in \mathcal{P}$ we may have many set of buyers $Y_t$.
			But, for all sets $Y_t$ and all prices $p \in \mathcal{P}$ it is true that there is no matching such that all buyers $Y_t$ can obtain one item.
			Notice that each buyer $i \in Y_t$ does not allocate any items at any price $p' < p$. Thus, each buyer $i \in Y_t$ is also
			not interested in any items sold before, otherwise by Lemma \ref{lem:gettingsomething} $i$ has to obtain something too.
			
			It means that all the items desired by buyers $Y_t$ are available at $p$. Thus, if it is not possible to match exactly one item to each buyer
			$i \in Y_t$, then also the optimal envy-free solution. does not allocate any items to buyers in $Y_t$.
			Since, this is true for each $p \in \mathcal{P}$, the lemma is proved.\qed
			
		\end{proof}
		
		Thus,
		
		\begin{align*}
		\mathcal{R}_{I_{3}}^{\mathbf{OPT}} \leq \sum_{i \in I_3^{>0}} b_i \leq  2 \mathcal{R}_{I_{3}}(\mathbf{X},\mathbf{p})
		\end{align*}
		
		where the first inequality is by Lemma \ref{lem:optdoesnotallocatetoo}. Last inequality is because to all $i \in I^{>0}_3$, whatever price $p$ is, the algorithm sells to them $\lfloor \frac{b_i}{p}\rfloor$ items. Thus the revenue extracted
		from them is 
		\begin{align*}
		\sum_{i \in I_3^{>0}} b_i
		\leq  \sum_{i \in I_3^{>0}} \bigg(\bigg\lfloor \frac{b_i}{p_i}\bigg\rfloor + 1\bigg) \cdot p_i
		\leq  \sum_{i \in I_3^{>0}} 2 \bigg\lfloor \frac{b_i}{p_i}\bigg\rfloor \cdot p_i \leq  2 \mathcal{R}_{I_{3}}(\mathbf{X},\mathbf{p})
		\end{align*}
		since $\lfloor \frac{b_i}{p}\rfloor \geq 1$ for all $i \in I^{>0}_3$.\qed
		
	\end{proof}
	
	Finally, by Lemma \ref{lem:I1}, Lemma \ref{lem:I2}, and Lemma \ref{lem:I3}:
	
	\begin{align*}
	\mathcal{R}^{\mathbf{OPT}}  &= \mathcal{R}_{I_{1}}^{\mathbf{OPT}} +\mathcal{R}_{I_{2}}^{\mathbf{OPT}} +\mathcal{R}_{I_{3}}^{\mathbf{OPT}} \\
	&2 \mathcal{R}_{I_{1}}(\mathbf{X},\mathbf{p})+2 \mathcal{R}(\mathbf{X},\mathbf{p})+2 \mathcal{R}_{I_{2}}(\mathbf{X},\mathbf{p})
	+2 \mathcal{R}_{I_{3}}(\mathbf{X},\mathbf{p}) \leq 4 \mathcal{R}(\mathbf{X},\mathbf{p})
	\end{align*}\qed
\end{proof}

\section{Limits of fixed-price auctions}\label{sec:lowerbound}
The first attempt to design revenue-maximizing envy-free algorithms is to consider the fixed-price scheme which assigns a uniform price for all items. The fixed-price scheme is practical  and is commonly used in the design of revenue-maximizing envy-free algorithms (see~\cite{DBLP:conf/wine/Colini-BaldeschiLSZ14,DBLP:conf/sigecom/FeldmanFLS12}). Our first result is that revenue obtained from the fixed-price scheme cannot approximate the optimal revenue within a factor of $O(\log n)$.  

\begin{theorem}\label{thm:lowerbound}
The optimal revenue of envy-free outcome in matching markets cannot be approximated within $O(\log n)$ by any fixed-price scheme.
\end{theorem}

\begin{proof}
	Consider an instance with an equal numbers of buyers and items. More specifically, let $I = \{1,\ldots, n\}$ and $J = \{j_1,\ldots, j_n\}$. 
	For each buyer $i \in I$, let $S_i = \{j_i\}$, and $v_i  = b_i = \frac{n}{n-i+1}$. As there is no intersection between the preference sets 
	of buyers, the optimal envy-free outcome allocates $S_i$ to buyer $i$ and charges buyer $i$ at price $\frac{n}{n-i+1}$.  Thus, $\mathcal{R}^{OPT} = \sum_{i=1}^n \frac{n}{n-i+1}$. Note that any fixed-price scheme decides an uniform price for all items. As the valuation and budgets are monotonic among all buyers, it implies that, when price $p$ is greater than $v_i$, only buyers with $k>i$ are allocated the items in their preference sets.   It is easy to  verify that the optimal revenue in the fixed-price scheme is $n$ by setting the uniform price equal to an arbitrary $v_i$. Hence,
	\[
	\frac{\mathcal{R}}{\mathcal{R}^{\mathbf{OPT}}} = \frac{v_{i} \cdot (n - i + 1)}{\sum_{i = 1}^n v_{i}} = \frac{n}{\sum_{i = 1}^n \frac{n}{n - (i- 1)}} \leq \frac{n}{n \cdot \log n} = \frac{1}{\log n}
	\]\qed
\end{proof}

\end{document}